\pdfoutput=1
\documentclass[10pt,pra,aps,superscriptaddress,nofootinbib,twocolumn,longbibliography]{revtex4-1}
\usepackage{amsmath,bbm}
\usepackage{amsthm}
\usepackage{amsmath}
\usepackage{latexsym}
\usepackage{amssymb}
\usepackage{float}
\usepackage{braket}
\usepackage{quantikz}
\usepackage{tikz}
\usepackage{graphicx}           
\usepackage{textcomp}
\usepackage{color}
\usepackage{xcolor}
\usepackage{mathpazo}
\usepackage{comment}
\usepackage{enumitem}
\usepackage{multirow}
\usepackage{setspace}
\usepackage{subcaption}
\usepackage[colorlinks=true,linkcolor=blue,citecolor=magenta,urlcolor=blue]{hyperref}
\usepackage{svg}
\usepackage{xr-hyper}   
\usepackage{hyperref}
\newcommand{\be}{\begin{equation}}
\newcommand{\ee}{\end{equation}}
\newcommand{\bea}{\begin{eqnarray}}
\newcommand{\eea}{\end{eqnarray}}
\def\squareforqed{\hbox{\rlap{$\sqcap$}$\sqcup$}}
\def\qed{\ifmmode\squareforqed\else{\unskip\nobreak\hfil
\penalty50\hskip1em\null\nobreak\hfil\squareforqed
\parfillskip=0pt\finalhyphendemerits=0\endgraf}\fi}
\def\endenv{\ifmmode\;\else{\unskip\nobreak\hfil
\penalty50\hskip1em\null\nobreak\hfil\;
\parfillskip=0pt\finalhyphendemerits=0\endgraf}\fi}

\newcommand{\tr}{\text{Tr}}
\newcommand{\I}{\mathbbm{1}}

\makeatletter
\newtheorem*{rep@theorem}{\rep@title}
\newcommand{\newreptheorem}[2]{%
\newenvironment{rep#1}[1]{%
 \def\rep@title{#2 \ref{##1}}%
 \begin{rep@theorem}}%
 {\end{rep@theorem}}}
\makeatother

\newtheorem{thm}{Theorem}
\newreptheorem{thm}{Theorem}
\newtheorem{lemma}{Lemma}

\begin{document}

\title{Limitation of maximally entangled probes for single-shot distinguishability of unitaries}

\author{Satyaki Manna}
\author{Anandamay Das Bhowmik}

\affiliation{School of Physics, Indian Institute of Science Education and Research Thiruvananthapuram, Kerala 695551, India}
\author{Debashis Saha}
\affiliation{School of Physics, Indian Institute of Science Education and Research Thiruvananthapuram, Kerala 695551, India}
\affiliation{Department of Physics, School of Basic Sciences, Indian Institute of Technology Bhubaneswar, Bhubaneswar, Odisha 752050, India}


\begin{abstract}
There have been many instances where the maximally entangled state as a probe acts better than the product and the non-maximally entangled states in the task of distinguishing quantum channels. We provide a proof that for single-shot discrimination of two unitary channels, entangled and product states are operationally equivalent. However, we identify pairs of unitaries that are perfectly distinguishable using a non-maximally entangled state, but not with a maximally entangled one. This contrast becomes more pronounced when the number of unitaries exceeds two. In every \textit{dimension $\geqslant 3$}, we show that there exists a class of unitaries that are indistinguishable under maximally entangled probes, yet perfectly distinguishable using product or non-maximally entangled inputs.  Another interesting set of unitaries in every \textit{dimension $\geqslant 3$} has been presented where only non-maximally entangled state acts as the successful probe, while product states and maximally entangled states cannot. 
\end{abstract}

\maketitle


\section{Introduction} \label{SEC I} 
Entanglement is the most unique and striking feature of quantum mechanics. Acting as a nonlocal resource, entanglement has been established as the predominant cornerstone of quantum information theory.  Most useful and well-acquainted applications of entanglement have been executed by utilizing the maximally entangled state.  Maximally entangled states play the most important role as the resource in several informational tasks such as quantum teleportation \cite{PhysRevLett.70.1895,Horodecki_1996,PhysRevLett.86.1370}, super-dense coding \cite{PhysRevLett.69.2881}, quantum cryptography \cite{PhysRevLett.67.661}, quantum error correction \cite{PhysRevA.101.042305}, quantum metrology \cite{Gerry:06} and quantum communication \cite{gisin2007quantum,PhysRevA.81.042326,PhysRevApplied.21.044010}.  Moreover, it is an important foundation of quantum computation and technology \cite{PhysRevA.100.022342,stasino2025implementation}. There has been another avenue of research exploiting the entanglement started when Kitaev \cite{AYuKitaev_1997} showed entanglement assisted probe-ancilla system sometimes gives the advantage in discriminating quantum channels. Over the years, the advantage of entanglement with respect to product probe system in the task of distinguishing quantum channels has been studied extensively \cite{PhysRevA.71.062340, PhysRevA.71.062310, PhysRevA.72.014305, PhysRevLett.102.250501, HORODECKI19961, PhysRevA.98.042103, npj, PhysRevA.111.022221, PhysRevA.81.032339, PhysRevLett.103.210501, PhysRevA.82.032302, watrous2008distinguishingquantumoperationshaving, Datta_2021}  but most of the literature did not compare the amount of entanglement needed for the necessary task.  At the moment, we are only interested in the discrimination of unitary channels \cite{GMauroD’Ariano_2002,PhysRevLett.87.177901, PhysRevA.64.050302,Ziman10022010}. Most of the works in 
this direction is being carried out for any two unitaries or a set of continuous unitaries.  For the estimation of unknown unitary from a continuous set of unitaries, maximally entangled state acts as the better or at par as the input state in reference to non-maximally entangled state \cite{PhysRevA.64.050302}. There are some works which consider multi-shot scenario of discrimination of quantum operations \cite{PhysRevLett.103.210501, PhysRevLett.87.177901}. We mainly address these questions concerning the hierarchy of the probing state, i.e.,  the maximally entangled state vs product state and maximally entangled state vs non-maximally entangled state and product state vs non-maximally entangled state in the distinguishability of finite number of unitaries in single shot paradigm. 

Distinguishability of different physical processes has been an important exploration to understand the intricacies of quantum theory and that enables us of a better understanding of the physical world. From the commence of this kind of research, distinguishability of quantum states \cite{Hellstrom, PhysRevA.93.062112, PhysRevA.70.022302, PhysRevA.88.052313} has been the most expansive area of investigation. Subsequently, the research on the distinguishability of quantum channels \cite{PhysRevLett.102.250501, PhysRevLett.87.177901, PhysRevA.71.062310, PhysRevA.72.014305, PhysRevA.64.050302,Ziman10022010, npj, PhysRevA.98.042103, PhysRevLett.103.210501, PhysRevA.81.032339, watrous2008distinguishingquantumoperationshaving} is now a very enriching and comprehensive field though it has more complex structure with respect to the distinguishability of states. In quantum foundation, this notion is extremely useful to comment on the reality of quantum states \cite{PhysRevLett.112.160404,PhysRevLett.112.250403,Chaturvedi2020quantum,PhysRevLett.113.020409, bhowmik2022interpretationquantumindistinguishability, ray2024epistemicmodelexplainantidistinguishability,chaturvedi2021}. Beyond the foundational insights, this quantity has several applications in the field of quantum information, quantum cryptography and quantum advantageous privacy preserving communication tasks \cite{PhysRevResearch.6.043269,PhysRevLett.115.030504,2025limits}.

In this paper, we formulate the problem of distinguishability of some a priori known set of unitary operations sampled from a probability distribution. In the single shot scenario, this problem of distinguishability can be classified into two types according to the initial input states, product or entangled. We have shown that the distinguishability of the unitaries ultimately reduces to the distinguishability of different evolved states depending on the probing state. Then for entanglement-assisted discrimination, we  show that all the maximally entangled states are equivalent in the task of distinguishing two unitaries.  Our study reveals that if any product input state provides the perfect distinguishability for two unitaries, then there exists at least one entangled state which can do the task at par. Then we check the hierarchy between the entangled states as the probe. We construct a very straight-forward necessary and sufficient criterion for which two unitaries are distinguishable with non-maximally entangled state but not distinguishable with maximally entangled state.  This phenomenon is not possible for two qubit unitaries. From this fact, we found an implication where any set of distinguishable qubit unitaries can be distinguishable with same probing state, i.e., maximally entangled state.   Then we produce our main result by presenting a set of $d$-dimensional $d$ unitaries which are distinguishable with product state  and non-maximally entangled state  but not distinguishable with any maximally entangled state.  Thereafter we demonstrate another set of $d$-dimensional $2d$ unitaries which are distinguishable with non-maximally entangled state but indistinguishable with both product and maximally entangled state. 

The paper is constructed as follows. In the next section, we present a general formulation of distinguishability of unitary operations. Then in the section \ref{result}, we exhibit our main results regarding the distinguishability of two unitaries and consequently, distinguishability of more than two unitaries. In conclusion, we
summarize the key findings and discuss several open
problems and potential future research directions.

\section{Distinguishability of Unitary Operations}
Unitary operator $(U)$ is a linear operator $U:H\rightarrow H$ on a Hilbert space $H$ that satisfies $U^\dagger U=UU^\dagger=\I$.

We consider a priori known set of $n$ unitaries $\{U_x\}_x$ acting on a $d$-dimensional quantum state and $x\in\{1,\cdots,n\}$. We assume that all the unitaries are sampled from a probability distribution $\{p_x\}_x$, i.e., $p_x>0$ and $\sum_x p_x=1$. To distinguish the unitaries, the unitary device is fed with a known quantum state, product or entangled and the device carries out one of $n$ unitary operations. After this process, the device gives a evolved state as the output. Therefore, one can perform any measurement on the evolved state and this measurement can be optimized such that the distinguishability of these evolved states will be maximum. Let us describe this optimal measurement by a set of POVM elements $\{N_{b}\}_b$, where $b\in\{1,\cdots,n\}$ and $\sum_b N_{b}=\I$. The protocol is successful in distinguishing the unitaries if $b$ is same as $x$. Any classical post processing of outcome $b$ can be included in the measurement $\{N_{b}\}_b$. Depending on the initial probing state, we can formulate two situations which are described in details in the next subsection.
\subsection{With Product System}
At first, the unitary device is given a product quantum state $\rho_A\otimes\rho_B$. After the device applying any of the unitaries from the set $\{U_x\}_x$ on part $A$, the output state will be $(U_x\rho_A U_x^\dagger)\otimes\rho_B$. After performing the measurement $\{N_b\}_b$, the distinguishability of the set of unitaries becomes the distinguishability of the set of evolved states. So distinguishability in this scenario, denoted by $\mathcal{D}_P$, is defined as,
\bea\label{D_S}
&&\mathcal{D}_P\left[\{U_x\}_x,\{p_x\}_x\right]\nonumber\\
&=& \max_{\rho_A} \sum_x p_xp(b=x|x)\nonumber\\
&=& \max_{{\rho_A},\{N_b\}}\sum_x p_x\tr((U_x\rho_A U_x^\dagger \otimes\rho_B)N_{b=x})\nonumber\\
&=& \max_{\rho_A} \mathcal{DS}[\{U_x\rho_A U_x^\dagger\}_x,\{p_x\}_x].\nonumber
\eea
$\mathcal{DS}[\{\rho_k\}_k,\{q_k\}_k]$ denotes the distinguishability of the set of states $\{\rho_k\}_k$, sampled from the probability distribution $\{q_k\}_k$. It is important to note that \eqref{D_S} does not depend on $\rho_B$.
This expression reduced as following by Hellstrom's formula \cite{Hellstrom} when we want to know the distinguishability of two unitaries $U_1$ and $U_2$:
\bea
&&\mathcal{D}_P\left[\{U_x\}_{x=1}^2,\{p_x\}_{x=1}^2\right]\nonumber\\
&=& \max_{\rho_A} \mathcal{DS}[\{U_x\rho_A U_x^\dagger\}_x,\{p_x\}_x].\nonumber\\
&=& \max_{\rho_A}\frac12\left(1+||p_1 U_1\rho_A U_1^\dagger-p_2 U_2\rho_A U_2^\dagger||\right).
\eea
If $U_1$ and $U_2$ are sampled from equal probability distribution, i.e., $p_1=p_2=1/2$, the above expression reduces to, 
\bea\label{D_S1}
&&\mathcal{D}_P\left[\{U_x\}_{x=1}^2,\{1/2,1/2\}\right]\nonumber\\
&=&\frac12+\frac14\max_{\rho_A}|| U_1\rho_A U_1^\dagger- U_2\rho_A U_2^\dagger||,
\eea
where $||.||$ (trace norm) denotes the sum of the absolute values of the eigenvalues. By convexity of trace norm, the maximum value of $|| U_1\rho_A U_1^\dagger- U_2\rho_A U_2^\dagger||$ is achieved when $\rho_A$ is the pure input state \cite{10.5555/3240076}. So for pure input state, i.e., $\rho_A=\ket{\psi}\bra{\psi}$, \eqref{D_S1} takes the form after implementing Hellstrom's formula as following:
\bea\label{D_S11}
&&\mathcal{D}_P=\max_{\ket{\psi}}\frac12\left[1+\sqrt{1-|\langle \psi| U_1^\dagger U_2|\psi|^2}\right].
\eea

$\mathcal{D}_P$ will be achieved when $|\bra{\psi}U_1^\dagger U_2\ket{\psi}|$ is minimum. As $U_1^\dagger U_2$ is an unitary, we can write the spectral decomposition of this operator as, $U_1^\dagger U_2=\sum_{j=1}^d e^{\mathbbm{i}\theta_j}\ket{\psi_j}\bra{\psi_j}$, where $e^{\mathbbm{i}\theta_j}$ are the eigenvalues of $U_1^\dagger U_2$ and $\ket{\psi_j}$ is the eigenvector corresponding to $j$'th eigenvalue. Therefore $\ket{\psi}$ can be decomposed into the linear combination of the eigenstates $\ket{\psi_j}$, i.e., $\ket{\psi}=\sum_{j=1}^d\alpha_j\ket{\psi_j} $. So, 
\bea\label{min_con}
\min_{\ket{\psi}}|\bra{\psi}U_1^\dagger U_2\ket{\psi}|^2&=& \min_{\ket{\psi}}|\bra{\psi}\sum_{j=1}^d e^{\mathbbm{i}\theta_j}\alpha_j\ket{\psi_j}|^2\nonumber\\
&=& \min_{\ket{\psi}}|\sum_{j=1}^d|\alpha_j|^2 e^{\mathbbm{i}\theta_j}|^2\nonumber\\
&=& \min|con\{e^{\mathbbm{i}\theta_j}\}|^2,
\eea
where $con\{e^{\mathbbm{i}\theta_j}\}$ denotes the set of complex numbers that can be written as the convex combinations of $\{e^{\mathbbm{i}\theta_j}\}$ and $\min|.|$ denotes the minimum norm over all those complex numbers (see figure \ref{figx}). Substituting \eqref{min_con} into \eqref{D_S11}, we get an elegant form,
\be
\mathcal{D}_P=\frac12\left[1+\sqrt{1-\min|con\{e^{\mathbbm{i}\theta_j}\}|^2}\right].
\ee
\subsection{With Entangled system}
Now, the unitary device is fed with a $d\otimes d'$ entangled state $\rho_{AB}$ and the device applies one of the unitaries from the set $\{U_x\}_x$ on the part $A$ of the entangled state. The evolved state will be $(U_x\otimes\I)\rho_{AB}(U_x^\dagger\otimes\I)$. Then one can perform the optimal measurement of dimension $dd'$. Similarly, the distinguishability of the unitaries in this scenario, denoted as $\mathcal{D}_E$, can be written as,
\bea\label{D_E}
&&\mathcal{D}_E\left[\{U_x\}_x,\{p_x\}_x\right]\nonumber\\
&=& \max_{\rho_{AB}} \sum_x p_xp(b=x|x)\nonumber\\
&=& \max_{\rho_{AB},\{N_b\}} \sum_x p_x \tr[(U_x\otimes\I)\rho_{AB}(U_x^\dagger\otimes\I)N_{b=x}]\nonumber\\
&=& max_{\rho_{AB}} \mathcal{DS}[\{(U_x\otimes\I)\rho_{AB}(U_x^\dagger\otimes\I)\}_x,\{p_x\}_x].
\eea

Now at this point, it is a natural and important question to be raised regarding the dimension of the system of part $B$. We prove that it suffices to take the ancillary system of dimension $d'=d$. The argument is given below. 
\begin{lemma}\label{suff_prob}
    For distinguishability in entanglement assisted scenario, the sufficient initial entangled state is of $d\otimes d$ dimension for $d$ dimensional unitaries.
\end{lemma}
\begin{proof}
    We can take the best possible entangled probing state $\ket{\psi}_{AB}\in\mathcal{C}^d\otimes\mathcal{C}^{d'}$. By Schmidt decomposition, we can always write the state $\ket{\psi}_{AB}=\sum_{l=1}^{d} C_l\ket{\eta_l}\ket{\chi_l}$, where $d'\geqslant d$. 
     From the structure of Schmidt decomposition, we can see extra $(d'-d)$ number of bases are redundant. The other case, when $d'<d$, the input state can be written as $\ket{\psi'}_{AB}=\sum_{l=1}^{d'} C_l\ket{\eta_l}\ket{\chi_l}$. So this state is not considerable because we can not write the state with all the bases of $\ket{\eta_l}$. To make this happen, we need to take at least $d=d'$.
\end{proof}
 Now we are moving to analyze some special kind of probing state to distinguish two unitaries in this entanglement assisted scenario. Our first venture to assess the hierarchy of different maximally entangled states as the probe.
 \begin{thm}\label{th4}
    The distinguishability of any two qudit unitaries sampled from uniform distribution in the entanglement-assisted scenario is invariant for any initial maximally entangled state.
\end{thm}
\begin{proof}
     Let us take the initial maximally entangled state as $\ket{\phi}=\frac{1}{\sqrt{d}}\sum_l\ket{ll}$.  From \eqref{D_E}, we can write,
    \bea\label{D_E_maximally_entangled}
   && \mathcal{D}_{ME}[\{U_1,U_2\},\{1/2,1/2\}]\nonumber\\
   &=&\mathcal{DS}[\{(U_1\otimes\I)\ket{\phi},(U_2\otimes\I)\ket{\phi}\},\{1/2,1/2\}]\nonumber\\
   &=&\frac12\left[1+\sqrt{1-|\bra{\phi}(U_1^\dagger\otimes\I)(U_2\otimes\I)\ket{\phi}|^2}\right]\nonumber\\
   &=& \frac12\left[1+\sqrt{1-\frac{1}{d^2}|\sum_{l=1}^d\bra{l}U_1^\dagger U_2\ket{l}|^2}\right]\nonumber\\
   &=& \frac12\left[1+\sqrt{1-\frac{1}{d^2}|\tr(U_1^\dagger U_2)|^2}\right].
    \eea
The third line comes from applying Hellstrom's formula \cite{Hellstrom}. All the maximally entangled states are connected via local unitaries. Now, we can start with another maximally entangled state which can be written as $\ket{\phi'}=(U\otimes\I)\ket{\phi}$.  As trace is basis independent property, $\tr(U_1^\dagger U_2)$ will be unchanged if $\ket{l}$ is changed to $U\ket{l}$. So \eqref{D_E_maximally_entangled} remains unchanged, and that completes the proof. 
\end{proof}  
Noting the fact that $\tr(U_1^\dagger U_2)$ is the sum of the eigenvalues of $U_1^\dagger U_2$, we can express the distinguishability of $U_1$ and $U_2$ with maximally entangled state $(\mathcal{D}_{ME})=\frac12\left[1+\sqrt{1-\frac{1}{d^2}|\sum_{j=1}^d e^{\mathbbm{i}\theta_j}|^2}\right]$. 

At figure \ref{figx}, we present a pictorial depiction of the position of the eigenvalues of $U_1^\dagger U_2$ in the Argand plane and then describe the distinguishability of $U_1$ and $U_2$ with product system probe ($\mathcal{D}_P$) and maximally entangled probe ($\mathcal{D}_{ME}$).

\begin{figure}[h!]
    \centering
    \includegraphics[width=0.7\linewidth]{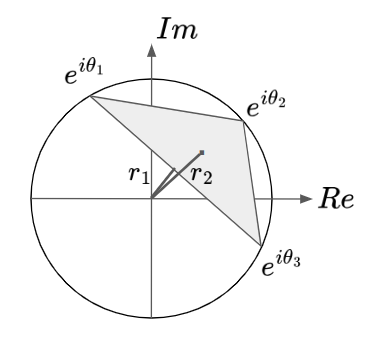}
    \caption{ An example is depicted, where the position of the eigenvalues of $U_1^\dagger U_2$ has been shown here in the complex plane and $U_1$ and $U_2$ are $3$-dimensional unitaries. $\mathcal{D}_P=\frac12(1+\sqrt{1-|r_1|^2})$ and $\mathcal{D}_{ME}=\frac12(1+\sqrt{1-|r_2|^2})$. In general, $\mathcal{D}_P\geqslant\mathcal{D}_{ME}$.}
    \label{figx}
\end{figure}

\section{Results}\label{result}

Before presenting our main claims of the disadvantageous property of  maximally  entangled states, we introduce some other results regarding the distinguishability of two unitaries. In this section, the unitaries we take at all the Theorems are sampled from a equal probability distribution. 
\begin{thm}\label{thm_2}
     For any two unitaries, if perfect distinguishability is achieved using product probing state, there exists at least one entangled probing state which gives same distinguishability.
\end{thm}
\begin{proof}
    Let us take two $d$-dimensional unitaries $U_1$ and $U_2$ (sampled from equal probability distribution) which are perfectly distinguishable. This implies that there exist $\{\alpha_j\}_{j=1}^d$ such that $|\sum_{j=1}^d|\alpha_j|^2 e^{i\theta_j}|=\min|con\{e^{i\theta_j}\}|=0$ (from \eqref{min_con}).

     If we start with an entangled probing state $\ket{\Psi}_{AB}=\sum_{l}\beta_{l}\ket{\psi_l}_A\ket{\psi_l}_B$,($\sum_l |\beta_l|^2=1$), the distinguishability of two unitaries will be,
     \bea\label{de_u1_u2_ent}
    &&\mathcal{D}_E\left[\{U_x\}_{x=1}^2,\{1/2,1/2\}\right]\nonumber\\
    &=& \frac12 \left(1 +\right.\nonumber\\
    &&\left. \sqrt{1-|\sum_{l}|\beta_{l}|^2\langle \psi_l\psi_l|(U_1^\dagger\otimes \I)(U_2\otimes\I) |\psi_l\psi_l\rangle|^2}\right)\nonumber\\
    &=& \frac12 \left(1+\sqrt{1-|\underbrace{\sum_{l}|\beta_{l}|^2\langle \psi_l|U_1^\dagger U_2 |\psi_l\rangle}_{I}|^2 }\right).
    \eea  

    We know $U_1^\dagger U_2=\sum_{j=1}^d e^{i\theta_j}\ket{\psi_j}\bra{\psi_j}$. We can rewrite the inner product term $I=|\sum_j|\beta_j|^2 e^{i\theta_j}|$. If we choose $\beta_j=\alpha_j, \forall j$, then $I=|\sum_j|\alpha_j|^2 e^{i\theta_j}|$, which is $\min|con\{e^{i\theta_j}\}|= 0$ from our initial claim. That completes the proof.

\end{proof}

One can think of the possibility of the inverse of the statement of Theorem \ref{thm_2} that is "if two unitaries are not distinguishable with product probing state, can the entangled probing state succeed to distinguish them?" The answer is negative and that can be inferred from the already proved result of \cite{GMauroD’Ariano_2002}. Now we move into a result regarding the hierarchy of entangled states as the probing state. 
\begin{thm}\label{max_nonmax}
    The necessary and sufficient conditions for two $d$-dimensional unitaries $U_1$ and $U_2$ to be not distinguishable with maximally entangled probing state but distinguishable with non-maximally probing state are following:\\
    $(i)$ $\tr(U_1^\dagger U_2)\neq 0$.\\
    $(ii)$ $\min|con\{e^{i\theta_j}\}|=0$, where $\{e^{i\theta_j}\}_{j=1}^d$ are the eigenvalues of $U_1^\dagger U_2$.
\end{thm}
\begin{proof}
 \eqref{D_E_maximally_entangled} gives the expression of distinguishability of two unitaries by using a maximally entangled state as the probing state.
For $\mathcal{D}_E[\{U_1,U_2\},\{1/2,1/2\}]$ not to be $1$, $\tr(U_1^\dagger U_2)$ must not be zero. 

From \eqref{min_con}, we can infer that two unitaries will be distinguishable by product system probe if $\min|con\{e^{i\theta_j}\}|=0$.  In this situation, Theorem \ref{thm_2} asserts that there exists at least one entangled state which can do the same task of perfect distinguishability.
\end{proof}
 
To cite an example, consider two $d$-dimensional $(d\geqslant 3)$ unitaries such that
\bea\label{nonmax_uni}
\tilde{U}_1&=&\sum_{j=1}^d\ket{\Psi_j}\bra{j}, \quad 
\tilde{U}_2=\sum_{j=1}^d\ket{\Psi^{'}_j}\bra{j},\nonumber\\
\text{with} &&\ket{\Psi_1}=\ket{\Psi^{'}_2} \text{  and  } \ket{\Psi_2}=\ket{\Psi^{'}_1}.
\eea

To distinguish this two unitaries non-maximally entangled state acts as a better probing state than maximally entangled state.

If we start with a non-maximally entangled state $\frac{1}{\sqrt{2}}\left(\ket{11}+\ket{22}\right)$, distinguishability of two unitaries described above will be the distinguishability of two following states: $\frac{1}{\sqrt{2}}\left(\ket{\Psi_1 1}+\ket{\Psi_2 2}\right)$ and $\frac{1}{\sqrt{2}}\left(\ket{\Psi^{'}_1 1}+\ket{\Psi^{'}_2 2}\right)$. These two states are orthogonal from the description of the unitaries at \eqref{nonmax_uni}.\\

Now if we take input state a maximally entangled one, i.e., $\frac{1}{\sqrt{d}}\sum_{j=1}^d\ket{jj}$, the evolved states will be $\sum_{j=1}^d\ket{\Psi_j}\ket{j}$ and $\sum_{j=1}^d\ket{\Psi^{'}_j}\ket{j}$. Using the Hellstrom's formula, we find the distinguishability of these two evolved states is $\frac12+\frac{\sqrt{d-1}}{d}$.
  At this point, a straight-forward Theorem can be concluded for the sets of distinguishable qubit unitaries from Theorem \ref{max_nonmax} and Theorem \ref{th4}.  
 \begin{thm}\label{set_u}
     Suppose there are $n$ number of sets $\{\mathbbm{S}_i\}_{i=1}^n$ and each set $\mathbbm{S}_i$ consists of distinguishable qubit unitaries. Every set is always distinguishable with a common maximally entangled probe.
 \end{thm}
 \begin{proof}
      The condition $\tr(U_1^\dagger U_2)=0$ and $\min|con\{e^{i\theta_j}\}|=0$ are equivalent for two unitaries acting on $\mathbbm{C}^2$, since $U_1^\dagger U_2$ has two eigenvalues. Thus if two unitaries are distinguishable, then they must be distinguishable with maximally entangled probe. From Theorem \ref{th4}, we know all the maximally entangled state is equivalent for distinguishing two unitaries. If number of unitaries exceeds than two, then maximally entangled state is only option trivially for perfect distinguishability as we always find at most one orthogonal state with respect to a qubit-qubit non-maximally entangled state and same goes for qubit product state also. So any maximally entangled state can work as a common probe for which the unitaries of the each set can be distinguishable.
 \end{proof}
 If we increase the dimension of the unitaries to $3$, then the Theorem \ref{set_u} does not hold. Next theorem is explaining that fact.\\
 \begin{thm}
     Suppose there are two sets ($\mathbbm{S'}_1,\mathbbm{S'}_2$) of qutrit distinguishable unitaries. $\{T_1, T_2\}\in\mathbbm{S'}_1$ and $\{T_1, T_3\}\in\mathbbm{S'}_2$ where
     \bea
     T_1 &=& \ket{1}\bra{1}+\ket{2}\bra{2}+\ket{3}\bra{3},\nonumber\\
     T_2 &=& \ket{1}\bra{1}+e^{\mathbbm{i}2\pi/3}\ket{2}\bra{2}+e^{\mathbbm{i}4\pi/3}\ket{3}\bra{3},\nonumber\\
     T_3 &=& \ket{1}\bra{1}+\ket{2}\bra{2}-\ket{3}\bra{3}.
     \eea
    There does not exist any common probe which can distinguish the unitaries of both the sets $\mathbbm{S'}_1$ and $\mathbbm{S'}_2$. 
 \end{thm}
 \begin{proof}

 Let us take a common probe $\ket{\zeta_{BB'}}=\sum_{k=1}^3\lambda_k\ket{\tau_k}\ket{k}$. It is easy to check that the eigenvectors of $T_1^\dagger T_2$ and $T_1^\dagger T_3$ are same $(\ket{1},\ket{2},\ket{3})$. We can write $\ket{\tau_k}=\sum_{j=1}^3 c_j^k\ket{j}$ and we can calculate 
 \bea
 IP_1 &=& \bra{\zeta_{BB'}}(T_1^\dagger T_2\otimes\I)\ket{\zeta_{BB'}}\nonumber\\
 &=& \sum_{k=1}^3|\lambda_k|^2\bra{\tau_k}T_1^\dagger T_2\ket{\tau_k}\nonumber\\
 &=& 1(R_1)+e^{\mathbbm{i}2\pi/3}(R_2)+e^{\mathbbm{i}4\pi/3}(R_3),\nonumber\\
 \text{where  } R_1 &=& \sum_{k=1}^3|\lambda_k|^2|c_1^k|^2,
 R_2 = \sum_{k=1}^3|\lambda_k|^2|c_2^k|^2, \nonumber\\
 R_3 &=& \sum_{k=1}^3|\lambda_k|^2|c_3^k|^2.
 \eea
  As $R_1+R_2+R_3 = 1$, $IP_1$ is a convex combination of $\{1,e^{\mathbbm{i}2\pi/3},e^{\mathbbm{i}4\pi/3}\}$. As these three points make a simplex structure, only one unique convex combination makes $IP_1$ zero ($IP_1=0$ means perfect distinguishability of $T_1$ and $T_2$) and that is $R_1=R_2=R_3=1/3$. We can do the similar calculation for $IP_2$.
 \bea
 IP_2 &=& \bra{\zeta_{BB'}}(T_1^\dagger T_3\otimes\I)\ket{\zeta_{BB'}}\nonumber\\
 &=&1(R_1)+1(R_2)-1(R_3). 
 \eea
 From the previous case, if $R_1=R_2=R_3=1/3$, this $IP_2\neq 0$. From Theorem \ref{thm_2}, we know if two unitaries are distinguishable with single system probe, there must exists an entangled probe which can perfectly distinguish these two unitaries. In this example, there is no common entangled probe which can distinguish both the sets of unitaries. On the other way, it can also be concluded that there does not exist any common single probe also to do the same task.     
 \end{proof}

 Now we will present our two main results by giving a general set of unitary operations in dimension $d\geqslant 3$. Both the theorems describe the disadvantage of maximally entangled state as a probing state. 

\begin{thm}\label{thm3}
    Consider $d$, $d$-dimensional unitaries $\{V_l=
    \sum_{j=1}^d\ket{\psi^{(l)}_j}\bra{j}\}_{l=1}^d$ and $\ket{\psi^{(1)}_j}=\ket{\psi^{(j)}_1}$ and $\ket{\psi^{(l)}_j}=\ket{\psi^{(l')}_j}$ for $j\neq l,l'$ and  $d\geqslant3$. The unitaries are distinguishable with product state and non-maximally entangled state but not distinguishable with maximally entangled state. 
\end{thm}
\begin{proof}
 At first, we will consider the distinguishability of the unitaries with product state. If the probing state is $\ket{\chi}_{max}=\ket{1}\ket{\overline{\chi}}$, the set of transformed states will be $\{\ket{\psi^{(l)}_1}\ket{\overline{\chi}}\}_{l=1}^d$ which is same as $\{\ket{\psi^{(1)}_j}\ket{\overline{\chi}}\}_{j=1}^d$ and this is a set of orthogonal states. So the unitaries are distinguishable with product system probe.

  For non-maximally entangled state, let us consider a Scmidt rank-$2$ entangled state $\ket{\phi_{pr}}=\sum_{t=1}^d(a_t\ket{1}+b_t\sum_{s=2}^d\ket{s})\ket{t}$ with $a_t\neq b_t$. Note that, $\ket{\phi_{pr}}$ may not be in Schmidt-decomposed form. With this probe, the distinguishability of two unitaries $V_l$ and $V_{l'}$ ( for any $l\neq l'$) is,
 \bea
 &&\mathcal{D}(V_l,V_{l'})\nonumber\\
 &=& \frac12\left[1+\sqrt{1-|\bra{\phi_{pr}}(V_l^\dagger V_{l'}\otimes\I)\ket{\phi_{pr}}|^2}\right]\nonumber\\
 &=&  \frac12\left[1+\sqrt{1-|\sum_{t=1}^d 2 Re(a^*_t b_t)+(d-2)|b_t|^2|^2}\right].\nonumber\\
 \eea
  We can always find suitable $a_t$ and $b_t$ to make $\sum_{t=1}^d 2 Re(a^*_t b_t)+(d-2)|b_t|^2=0$, which means each pair of unitaries are distinguishable with same probe. 

 Now we will check about the maximally entangled probe. We consider the first two unitaries which are denoted as $\{V_l = \sum_{j=1}^d\ket{\psi^{(l)}_j}\bra{j}\}_{l=1}^2$. From Theorem \ref{th4}, the distinguishability with maximally entangled probe depends on $\tr(V_1^\dagger V_2)$. One can check $\tr(V_1^\dagger V_2)= d-2$, which is non-zero as we take $d\geqslant 3$. From \eqref{D_E_maximally_entangled}, we can infer that maximally entangled state can not distinguish the set of unitaries perfectly. That completes the proof. 
\end{proof}

    Let us give an example of $d$, $d$-dimensional unitaries as following which obey Theorem \ref{thm3} :\\
\begin{widetext}
    \bea\label{eq12}
    \Tilde{V}_1 &=& \ket{1}\bra{1}+\ket{2}\bra{2}+\ket{3}\bra{3}+\ket{4}\bra{4}+\ket{5}\bra{5}+\cdots+\ket{d-1}\bra{d-1}+\ket{d}\bra{d}\nonumber\\
     \Tilde{V}_2 &=& \ket{2}\bra{1}+\ket{1}\bra{2}+\ket{3}\bra{3}+\ket{4}\bra{4}+\ket{5}\bra{5}+\cdots+\ket{d-1}\bra{d-1}+\ket{d}\bra{d}\nonumber\\
     \Tilde{V}_3 &=& \ket{3}\bra{1}+\ket{2}\bra{2}+\ket{1}\bra{3}+\ket{4}\bra{4}+\ket{5}\bra{5}+\cdots+\ket{d-1}\bra{d-1}+\ket{d}\bra{d}\nonumber\\
     \Tilde{V}_4 &=& \ket{4}\bra{1}+\ket{2}\bra{2}+\ket{3}\bra{3}+\ket{1}\bra{4}+\ket{5}\bra{5}+\cdots+\ket{d-1}\bra{d-1}+\ket{d}\bra{d}\nonumber\\
    & \vdots &
 \hspace{50pt} \vdots  \hspace{50pt} \vdots  \hspace{50pt} \vdots  \hspace{50pt} \vdots
 \hspace{50pt} \vdots \nonumber\\
     \Tilde{V}_{d-1} &=& \ket{d-1}\bra{1}+\ket{2}\bra{2}+\ket{3}\bra{3}+\ket{4}\bra{4}+\ket{5}\bra{5}+\cdots+\ket{1}\bra{d-1}+\ket{d}\bra{d}\nonumber\\
     \Tilde{V}_d &=& \ket{d}\bra{1}+\ket{2}\bra{2}+\ket{3}\bra{3}+\ket{4}\bra{4}+\ket{5}\bra{5}+\cdots+\ket{d-1}\bra{d-1}+\ket{1}\bra{d}
    \eea
\end{widetext}
 In the table \ref{mytable}, we present the values of distinguishability of the unitaries decribed at \eqref{eq12} with product system $(\mathcal{\Tilde{D}}_P)$, with non-maximally entangled system $(\mathcal{\Tilde{D}}_{NME})$  and with maximally entangled probe $(\mathcal{\Tilde{D}}_{ME})$ for dimension ($d$)$=\{2,\cdots,7\}$ using semi-definite programming. Note that the dimension ($d$) of the unitaries is the same as the number of the unitaries.

\begin{table}[h!]
\centering
\begin{tabular}{|c|c|c|c|c|c|c|} 
 \hline
 $d$ & 2 & 3 & 4 & 5 & 6 & 7   \\ [0.6ex] 
 \hline
 $\mathcal{\Tilde{D}}_P$ & 1 & 1 & 1 & 1 & 1 & 1 \\ [0.6ex]
 \hline
 $\mathcal{\Tilde{D}}_{NME}$ & 1 & 1 & 1 & 1 & 1 & 1 \\ [0.6ex]
 \hline
 $\mathcal{\Tilde{D}}_{ME}$ & 1 & 0.9605 & 0.8980 & 0.8285 & 0.7616 & 0.7008 \\ [0.6ex]
 \hline
\end{tabular}
\caption{ Values of $\mathcal{\Tilde{D}}_P$ and  $\mathcal{\Tilde{D}}_{NME}$ are $1$ for all $d$ but the value of $\mathcal{\Tilde{D}}_{ME}$ is decreasing with increasing $d$.}
\label{mytable}
\end{table}

\begin{thm}\label{th_6}
    Consider $2d, d$-dimensional unitaries $\{W_k =\sum_{i=1}^d\ket{\phi_{i+k-1}}\bra{i}, W_{d+k}= -\ket{\phi_{k}}\bra{1}+\sum_{i=2}^d\ket{\phi_{i+k-1}}\bra{i}\}_{k=1}^d$ with $\ket{\phi_{d+x}}=\ket{\phi_x}$ and $ d\geqslant 3$. These unitaries are not distinguishable with both the product state and maximally entangled state but distinguishable with non-maximally entangled probe .
\end{thm}
\begin{proof}
    Probing with a product state is equivalent to probing with a single system as the auxiliary system does not contribute in the inner product of Hellstrom formula \eqref{D_S11}. So the product probe eventually transforms into $2d$ number of evolved states spanned in dimension $d$. So distinguishability is less than $1$ using the product state.

    For maximally entangled state, we take a pair of unitaries such as $W_1=\sum_{i=1}^d\ket{\phi_{i}}\bra{i}, W_{d+1}= -\ket{\phi_{1}}\bra{1}+\sum_{i=2}^d\ket{\phi_{i}}\bra{i}$ from the set. $\tr(W_1^\dagger W_2)=d-2$, which is non-zero for $d\geqslant 3$. Using Theorem \ref{th4}, we can say maximally entangled state cannot distinguish the set of unitaries perfectly.

    Now we probe with non-maximally entangled state $\ket{\psi_{nm}}=\sum_{i=1}^d\epsilon_i\ket{\phi_i}\ket{i}$ with the condition $-|\epsilon_1|^2+\sum_{i=2}^d|\epsilon_i|^2=0$. This condition, together with normalization, immediately gives $|\epsilon_1|=\frac{1}{\sqrt{2}}$ and it follows that $\sum_{i=2}^d|\epsilon_i|^2=\frac12$.  With this input state, the evolved states will be $\{\ket{\rho_k}=\sum_{i=1}^d\epsilon_i\ket{\phi_{i+k-1}}\ket{i},\ket{\rho_{d+k}}= -\epsilon_1\ket{\phi_{k}}\ket{1}+\sum_{i=2}^d\epsilon_i\ket{\phi_{i+k-1}}\ket{i}\}_{k=1}^d$. If we take inner product of any pair of states from this set, there will be three types of terms, such as:
    \bea
    \langle\rho_{k'}|\rho_k\rangle_{k\neq k'} &=& \sum_{i=1}^d|\epsilon_i|^2\langle \phi_{i+k'-1}|\phi_{i+k-1}\rangle\nonumber\\ &=& 0,\\
    \langle\rho_{d+k'}|\rho_{d+k}\rangle_{k\neq k'} &=& -|\epsilon_1|^2\langle\phi_{k'}|\phi_k\rangle\nonumber\\ &&+ \sum_{i=2}^d|\epsilon_i|^2\langle \phi_{i+k'-1}|\phi_{i+k-1}\rangle\nonumber\\
    &=& 0\\
    \langle\rho_{k'}|\rho_{d+k}\rangle &=& -|\epsilon_1|^2\langle\phi_{k'}|\phi_k\rangle\nonumber\\ &&+ \sum_{i=2}^d|\epsilon_i|^2\langle \phi_{i+k'-1}|\phi_{i+k-1}\rangle\nonumber\\
    &=& (-|\epsilon_1|^2+\sum_{i=2}^d|\epsilon_i|^2)\delta_{k,k'}\nonumber\\
    &=& 0.
    \eea
     The last inner product becomes zero from the initial condition imposed on $\ket{\psi_{nm}}$ when $k=k'$.
    So evolved states are mutually orthogonal. Thus the unitaries are distinguishable with non-maximally entangled probe. Note that, the highest Schmidt rank of $\ket{\psi_{nm}}$ can be $d$ and the lowest rank can be $2$. 
\end{proof}
Let us give an example of $6$ unitaries of dimension $3$ as following which abide by Theorem \ref{th_6}:\\
\bea
\Tilde{W}_1 &=& \ket{1}\bra{1}+\ket{2}\bra{2}+\ket{3}\bra{3}\nonumber\\
\Tilde{W}_2 &=& \ket{2}\bra{1}+\ket{3}\bra{2}+\ket{1}\bra{3}\nonumber\\
\Tilde{W}_3 &=& \ket{3}\bra{1}+\ket{1}\bra{2}+\ket{2}\bra{3}\nonumber\\
\Tilde{W}_4 &=& -\ket{1}\bra{1}+\ket{2}\bra{2}+\ket{3}\bra{3}\nonumber\\
\Tilde{W}_5 &=& -\ket{2}\bra{1}+\ket{3}\bra{2}+\ket{1}\bra{3}\nonumber\\
\Tilde{W}_6 &=& -\ket{3}\bra{1}+\ket{1}\bra{2}+\ket{2}\bra{3}
\eea
These six unitaries are distinguishable with initial probe $\frac{1}{\sqrt{2}}\ket{11}+c_1\ket{22}+c_2\ket{33}$ but not distinguishable with any input maximally entangled state or product state.

 In the table \ref{mytable1}, we present the values of distinguishability of the unitaries described at Theorem \ref{th_6} with product system $(\mathcal{\overline{D}}_P)$, with non-maximally entangled system $(\mathcal{\overline{D}}_{NME})$ and with maximally entangled system $(\mathcal{\overline{D}}_{ME})$ for dimension ($d$)$=\{2,\cdots,6\}$ using semi-definite programming.
\begin{table}[h!]
\centering
\begin{tabular}{|c|c|c|c|c|c|c|} 
 \hline
 $d$ & 2 & 3 & 4 & 5 & 6  \\ [0.6ex] 
 \hline
 $\mathcal{\overline{D}}_P$ & 1/2 & 1/2 & 1/2 & 1/2 & 1/2 \\ [0.6ex]
 \hline
 $\mathcal{\overline{D}}_{NME}$ & 1 & 1 & 1 & 1 & 1 \\ [0.6ex]
 \hline
 $\mathcal{\overline{D}}_{ME}$ & 1 & 0.9146 & 0.8163 & 0.7300 & 0.6570 \\ [0.6ex]
 \hline
\end{tabular}
\caption{ Values of $\mathcal{\overline{D}}_P$ is $1/2$ for all dimensions and  $\mathcal{\overline{D}}_{NME}$ are $1$ for all $d$ but the value of $\mathcal{\overline{D}}_{ME}$ is decreasing with increasing $d$.}
\label{mytable1}
\end{table}

\section{Conclusion}
In this paper, our principal result is to provide two sets of unitaries, where first set consists of the unitaries which are distinguishable with  non-maximally entangled state and product state  as the input but indistinguishable with  maximally entangled  input state and  the unitaries of the second set are distinguishable with non-maximally entangled state but not distinguishable with product state and maximally entangled state.
 Additionally we give a necessary and sufficient conditions of superiority of non-maximally entangled state over maximally entangled state as the probe to distinguish two unitaries.   We expect that our work encourages further exploration about the role of entanglement in the task of distinguishing physical processes. 

Other future works could explore a more general classes of unitaries to satisfy our core claim and eventually find some necessary and sufficient conditions which will characterize this phenomenon. Furthermore, one can investigate our results in the context of some general quantum channels. In our study, it has been shown that the input entangled states (maximally entangled or non-maximally entangled) cannot be worse than the input product states when we consider the perfect distinguishability of only two unitaries. It would be an interesting attempt to check this nature of entangled states for any two general qubit channels and this venture can lead to the more fundamental question of optimal input state to distinguish two or more channels. If not for any general channels, it will be very beneficial if we can comment about the input states for any special kind of channels, specially when these channels are not distinguishable perfectly. On top of that, one can think of the applications of the distinguishability of unitaries in the different kinds of information and computation theoretic tasks \cite{HUANG2022127863,PhysRevA.106.062429}.

\subsection*{Acknowledgements}
This work is supported by STARS (STARS/STARS-2/2023-0809), Govt.
of India.

\bibliography{ref} 

\end{document}